\documentclass[12pt,a4paper,english]{amsart}
\usepackage[utf8]{inputenc}
\usepackage[english]{babel} 

\usepackage{amsmath} 
\usepackage{amsthm} 
\usepackage{graphicx} 
\usepackage{xcolor} 
\usepackage{amsfonts}
\usepackage[biblabel]{cite}
\usepackage{bm} 
\usepackage[hidelinks]{hyperref} 
\usepackage{amssymb} 
\usepackage{amscd}
\usepackage{etoolbox}
\patchcmd{\thmhead}{(#3)}{#3}{}{}
\usepackage{enumitem}
\usepackage{breqn} 
\usepackage{faktor} 
\usepackage{xfrac} 
\usepackage{standalone}
\usepackage{tikz}

\usepackage{mathtools} 
\DeclarePairedDelimiter\abs{\lvert}{\rvert}%
\DeclarePairedDelimiter\norm{\lVert}{\rVert}%
\makeatletter
\let\oldabs\abs
\def\abs{\@ifstar{\oldabs}{\oldabs*}}
\let\oldnorm\norm
\def\norm{\@ifstar{\oldnorm}{\oldnorm*}}
\makeatother

\newtheorem{theorem}{Theorem}[section]
\newtheorem{proposition}[theorem]{Proposition}
\newtheorem{corollary}[theorem]{Corollary}
\newtheorem{lemma}[theorem]{Lemma}
\theoremstyle{definition}
\newtheorem{definition}[theorem]{Definition} 
\newtheorem{remark}[theorem]{Remark} 
\newtheorem{example}[theorem]{Example}

\DeclareMathOperator{\FB}{FB}
\DeclareMathOperator{\AG}{AG}
\DeclareMathOperator{\FR}{FR}

\DeclareMathOperator{\wt}{{wt}}
\DeclareMathOperator{\ev}{{ev}}
\DeclareMathOperator{\ini}{{in}}
\DeclareMathOperator{\supp}{{supp}}

\newcommand{\vladut}{Vl\u{a}du\c{t}}
\newcommand{\F}{{\mathbb{F}}}
\newcommand{\fq}{\mathbb{F}_q}
\newcommand{\NN}{{\mathbb{N}}}

\newcommand{\M}{{\mathbb{M}}}
\newcommand{\X}{{\mathbb{X}}}

\newcommand{\bqm}{{[0,q-1]^m}}

\newcommand{\set}[1]{\left\{#1\right\}}
\newcommand{\size}[1]{\left|#1\right|}
\newcommand{\mon}{\text{Mon}}

\title[Wei-type duality and asymptotics of the footprint bound]{Wei-type duality and asymptotics of the footprint bound}
\author{Rodrigo San-José}
\curraddr{
\texttt{Department of Mathematics\\ Virginia Tech\\ Blacksburg, VA USA}
}
\email{rsanjose@vt.edu}

\thanks{The author was partially supported by the NSF grant DMS-2401558, the Commonwealth Cyber Initiative, an AMS-Simons Travel Grant, and by Grant PID2022-138906NB-C21 funded by MICIU/AEI/10.13039/501100011033 and by ERDF/EU}



\begin{document}

\begin{abstract}
We obtain a Wei-type duality between the footprint bound and the dual footprint bound for the generalized Hamming weights of an evaluation code. This duality applies between the Andersen-Geil and Feng-Rao bounds as well. We also prove that the footprint and dual footprint bounds cannot be used to guarantee the asymptotic goodness of a family of evaluation codes. 
\end{abstract}
\maketitle

\section{Introduction}
The footprint bound is a classical result in algebraic geometry, bounding the number of rational points defined by an ideal in terms of its footprint \cite[Prop. 7, Chapter 5 \textsection 3]{cox}. This idea has proven fruitful for bounding the minimum distance of linear codes \cite{geilbezout,pitones_minimum_distance_functions,sanjose_conjecture}, since it is closely related to the number of zeroes of polynomials over a finite field. The generalized Hamming weights (GHWs) of a linear code, introduced in \cite{weiGHW}, provide an extension of the minimum distance, which has found several applications \cite{guruswammiGHWlistdecoding,guruswammiGHWlistdecodingTensorInterleaved}. Moreover, it is related to the number of common zeroes of sets of polynomials and can therefore also be studied using the footprint bound. Among other properties, the GHWs of a linear code satisfy the so-called Wei duality, which implies that the GHWs of a linear code are determined by those of its dual, and vice versa. A further extension of this concept is given by the relative generalized Hamming weights (RGHWs) of a pair of linear codes \cite{luoPropertiesRGHWs}, which have applications in secret sharing \cite{matsumotoRGHW} and quantum error-correction \cite{kkks,hamadaSteaneEnlargement}. There are two related bounds, the Feng-Rao bound \cite{fengraoMajority}, and the Andersen-Geil bound (sometimes called the Feng-Rao bound for primary codes) \cite{andersen_geil_bound}, which generalize the footprint bound in the affine setting and can also be used to bound the RGHWs of linear codes \cite{geil_rghws_one_point_ag_codes}. The computation of the GHWs and RGHWs of linear codes is, in general, NP-hard \cite{vardyIntractability}, and the footprint and Feng-Rao-type bounds have been used to obtain them for many of the most well-known families of codes \cite{pellikaanGHWRM,beelenGHWcartesian,dattaRGHWcartesian,munueraGHWhermitica,sanjoseGHWNT,sanjose_squarefree,eduardoGHWHyperbolic,geil_rghws_rm,geil_rghws_one_point_ag_codes}.

Finding constructions for asymptotically good families of codes has proven a challenging problem. In fact, many of the most well-known algebraic families of codes are known to be asymptotically bad, e.g., Reed-Muller codes, Cartesian codes, hyperbolic codes \cite{olav_codes_from_order_domains}, binary primitive narrow-sense BCH codes \cite{lin_bch_bad}, and several classes of cyclic codes \cite{berman_cyclic_bad,conchita_cyclic_bad}. One of the most important examples of asymptotically good families of codes is given by AG codes \cite{tsfasman_asymptotically_good_ag_codes,garcia_stichtenoth_tower}. The authors consider the Goppa bound to prove that the corresponding codes are asymptotically good. Since the Feng-Rao and Andersen-Geil bounds are stronger, these bounds can be used to certify asymptotic goodness, in the sense that one can prove that certain families of codes are asymptotically good using them instead of the actual minimum distance of the code. 

In this paper, we study the footprint bound and the dual footprint bound from a combinatorial perspective. In particular, in Section \ref{s:wei_duality} we prove that there is a Wei-type duality between the values of the footprint bound and the dual footprint bound, and, more generally, between the values of the Andersen-Geil bound and the Feng-Rao bound. This implies that the footprint bound is sharp if and only if the dual footprint bound is sharp. Moreover, we show that the dual footprint bound coincides with the footprint bound of the dual code whenever the dual code is also obtained as a monomial code via an order-reversing bijection. This covers several well-known families of codes, such as Cartesian codes \cite{lopez_affine_cartesian} and decreasing norm-trace codes \cite{decreasingnormtrace}. These results also apply analogously to the Andersen-Geil and Feng-Rao bounds. 

In Section \ref{s:asymptotic}, we prove that the footprint bound cannot be used to guarantee that a family of codes is asymptotically good, i.e., if we have a family of codes with non-vanishing asymptotic rate, the quotient of the footprint bound of the corresponding codes, divided by the length, will vanish. Furthermore, we show a similar result for the footprint bound for GHWs and RGHWs, and the result also holds for the dual footprint bound. As a consequence, any family of evaluation codes for which the footprint bound is sharp is asymptotically bad, recovering the well-known results for Reed-Muller codes, Cartesian codes, or hyperbolic codes \cite{olav_codes_from_order_domains}. 

\section{Preliminaries}\label{s:preliminaries}
Let $\F_q$ be a finite field with $q$ elements. A linear code $C$ is a linear subspace of $\fq^n$. Its dual code is the orthogonal space with respect to the Euclidean inner product, and is denoted by $C^\perp$. Given a vector $c\in \fq^n$, we denote its Hamming weight by $\wt(c):=\size{\set{i:c_i\neq 0 }}$, which is the number of nonzero coordinates of $c$. We say that $C$ is an $[n,k,d]$ code if $C\subset \fq^n$, $\dim C=k$, and $d:=\min \set{\wt(c):c\in C\setminus \set{0}}$. The parameters $n, k, d$ are called the length, dimension, and minimum distance, respectively. Given a subcode $D\subset C$, that is, a linear subspace of $C$, we denote
$$
\supp(D):=\set{1\leq i \leq n:\exists c\in D \text{ with } c_i\neq 0}. 
$$
\begin{definition}
Let $1\leq r \leq k=\dim C$. The $r$-th generalized Hamming weight (GHW) of $C$ is a generalization of the minimum distance introduced in \cite{weiGHW}, and is defined by
$$
d_r(C):=\min \left\{ \abs{ \supp(D) }: D \textnormal{ is a subcode of } C \textnormal{ of dimension } r \right\}. 
$$
\end{definition}
It follows from the definition that $d_1(C)$ is the minimum distance of $C$. From \cite{weiGHW} we have the following general properties of the GHWs of a code.

\begin{theorem}[(Monotonicity)]\label{t:monotonia}
For an $[n,k]$ linear code $C$ with $k>0$ we have
$$
1\leq d_1(C)<d_2(C)<\cdots <d_k(C)\leq n.
$$
\end{theorem}

\begin{theorem}[(Duality)]\label{t:ghwdual}
Let $C$ be an $[n,k]$ code. Then
$$
\{d_r(C):1\leq r\leq k\}\sqcup \{n+1-d_r(C^\perp):1\leq r\leq n-k\}=\{1,2,\dots,n\}.
$$
\end{theorem}
The relative generalized Hamming weights (RGHWs) were introduced in \cite{luoPropertiesRGHWs}, and they extend the notion of GHWs. 

\begin{definition} \label{d:rghw}
Let $C_2\subset C_1 \subset \fq^n$ be two linear codes, and $k_1=\dim C_1$, $k_2=\dim C_2$. Let $r$ with $1\leq r \leq k_1-k_2$. The $r$-th relative generalized Hamming weight of $C_1$ and $C_2$, denoted by $M_r(C_1,C_2)$, is
$$
M_r(C_1,C_2)=\min \{ \abs{\supp(D)} : D \text{ is a subcode of $C_1$ with } \dim D=r, \; D\cap C_2=\{0\}\}.
$$
\end{definition}
For $1\leq r \leq k_1-k_2$, we have
$$
M_r(C_1,C_2)\geq d_r(C_1).
$$
If $C_2=\set{0}$, we recover the usual GHWs. The RGHWs are also strictly increasing \cite{luoPropertiesRGHWs}, although we do not have a duality result as in Theorem \ref{t:ghwdual} (see \cite[Ex. 2.10]{sanjoseGHWsPackage}). 

\begin{theorem}\label{t:monotoniarghw}
Let $C_2\subset C_1\subset \fq^n$ be linear codes with $\dim C_1=k_1$, $\dim C_2=k_2$. Then
$$
1\leq M_1(C_1,C_2)< M_2(C_1,C_2)<\cdots < M_{k_1-k_2}(C_1,C_2)\leq n.
$$
\end{theorem}

While it is possible to derive the GHWs and RGHWs of some linear codes directly from the definitions, e.g., see \cite{weiGHW} for binary Reed-Muller codes or \cite{sanjoseGHWMPC} for matrix-product codes, we usually require additional structure to compute them. In particular, in this work we consider evaluation codes. Let $\X=\set{P_1,\dots,P_n}\subset \fq^m$ be a set of $n$ points. We denote its vanishing ideal by $I(\X)$, which is the ideal generated by the polynomials in $\fq[x_1,\dots,x_m]$ that vanish at all the points of $\X$. We define the \textit{evaluation map} by 
$$
\begin{array}{lccc}
{\rm ev_\X}\colon &\fq[x_1,\dots,x_m] &\rightarrow& \fq^{n}\quad \\
&f & \mapsto& \left(f(P_1),\ldots,f(P_n)\right).
\end{array}
$$
We may denote it by $\ev$ if there is no confusion about the evaluation points. We denote by $\mon$ the set of monomials of $\fq[x_1,\dots,x_m]$. Given a monomial order and any ideal $I\subset \fq[x_1,\dots,x_m]$, we consider $\ini(I)$, the initial ideal of $I$ with respect to that order. The \textit{footprint} of $I$ is 
$$
\Delta(I):=\set{x^\alpha \in \mon : x^\alpha \not \in \ini(I)}. 
$$
The set $\Delta(I)$ is a basis for $\fq[x_1,\dots,x_m]/I$ \cite[Thm. 15.3]{eisenbud}. Consider $\mathbb{L}\subset \text{Span}_\fq(\Delta(I(\X)))$. Then the \textit{evaluation code} associated to $\X$ and $\mathbb{L}$ is $\ev(\mathbb{L})$. Let $\M\subset \Delta(I(\X))$. If $\mathbb{L}=\text{Span}_\fq(\M)$, we say that $\ev(\mathbb{L})$ is a \textit{monomial code}, and we denote it by $C(\X,\M)$. 

Let $V(I)$ be the affine variety associated to $I$, i.e., the set of all common zeroes of the polynomials in $I$. The \textit{footprint bound} states that
\begin{equation}\label{eq:footprint_bound}
\size{V(I)}\leq \size{\Delta(I)},
\end{equation}
with equality if $I$ is radical (see \cite[Thm. 6 and Prop. 7, Chapter 5 \textsection 3]{cox}). Let $f\in \mathbb{L}$. Using \eqref{eq:footprint_bound}, we get
\begin{equation}\label{eq:proof_footprint}
\size{V_\X(f)}:=\abs{V(f)\cap \X}= \abs{\Delta(I(\X)+(f))}\leq \size{\Delta(I(\X))\cap \Delta(f)}=\size{\Delta(I(\X))\cap \Delta(\ini(f))}. 
\end{equation}
This can be directly used to lower-bound the minimum distance of evaluation codes:
$$
d_1(\ev(\mathbb{L}))\geq n-\max\set{\size{\Delta(I(\X))\cap \Delta(x^\alpha)}:x^\alpha \in \ini(\mathbb{L})},
$$
where $\ini(\mathbb{L}):=\set{\ini(f):f\in \mathbb{L}}$. If we consider instead a linearly independent set $F=\set{f_1,\dots,f_r}\subset \mathbb{L}$, and $V_\X(F)=V(F)\cap \X$, the same reasoning as \eqref{eq:proof_footprint} gives a bound for the GHWs: 
\begin{equation}\label{eq:ghw_footprint_bound}
\begin{aligned}
d_r(\ev(\mathbb{L}))&\geq n-\max\set{\size{\Delta(I(\X))\cap \Delta(\mathbb{S})}:\mathbb{S}\in \binom{\ini(\mathbb{L})}{r}}\\
&\geq\min\set{\size{\Delta(I(\X))\setminus \Delta(\mathbb{S})}:\mathbb{S}\in \binom{\ini(\mathbb{L})}{r}},
\end{aligned}
\end{equation}
where we have used the fact that $\size{\Delta(I(\X))}=n$, which follows from \eqref{eq:footprint_bound}, and the notation $\binom{A}{r}$ for the subsets of $A$ with size $r$. This is what we call \textit{footprint bound} in the context of coding theory. We can interpret \eqref{eq:ghw_footprint_bound} combinatorially. For a set of monomials $A\subset \fq[x_1,\dots,x_m]$, we consider $\varphi(A):=\set{\alpha : x^\alpha \in A}$. Then we denote $X:=\varphi( \Delta(I(\X)))$ and $L:=\varphi(\ini(\mathbb{L}))$. For $\beta \in \mathbb{N}^m$, we denote $\nabla_X(\beta):=\set{\alpha \in X:\beta \preceq \alpha }$, where $\prec$ denotes the usual partial order in $\mathbb{N}^m$. The set $\nabla_X(\beta)$ is called the (upward) \textit{shadow} of $\beta$ (with respect to $X$). For a subset $S\subset \mathbb{N}^m$, we denote $\nabla_X(S):=\bigcup_{s\in S}\nabla_X(s)$. Then \eqref{eq:ghw_footprint_bound} is translated to
\begin{equation}\label{eq:footprint_bound_combinatorial}
d_r(\ev(\mathbb{L}))\geq \FB_X^r(L):=\min\set{ \size{\nabla_X(S)}:S\in \binom{L}{r}}. 
\end{equation}

\begin{example}
The bound from \eqref{eq:footprint_bound_combinatorial} is sharp for Reed-Muller codes \cite{pellikaanGHWRM}, Cartesian codes \cite{beelenGHWcartesian}, hyperbolic codes \cite{eduardoGHWHyperbolic}, codes over simplices \cite{sanjose_simplex}, and square-free Cartesian codes \cite{sanjose_squarefree}. It is also known not to be sharp for some AG codes; see, e.g., \cite{sanjoseGHWNT}.
\end{example}

In \cite{olav_evaluation_codes_affine_variety}, a similar bound for the GHWs of $\ev(\mathbb{L})^\perp$ is derived:
\begin{equation}\label{eq:footprint_bound_combinatorial_dual}
d_r(\ev(\mathbb{L})^\perp)\geq  \FB_X^{r,\perp}(L):=\min\set{ \size{\nabla^\perp_X(S)}:S\in \binom{X\setminus L}{r}},
\end{equation}
where $\nabla^\perp_X(\beta):=\set{\alpha \in X:\alpha \preceq \beta }$ and $\nabla^\perp_X(S):=\bigcup_{s\in S}\nabla_X^\perp(s)$. We call this bound \textit{dual footprint bound}. 

\begin{remark}\label{r:evalcodes_to_monomial}
Note that the bounds \eqref{eq:footprint_bound_combinatorial} and \eqref{eq:footprint_bound_combinatorial_dual} for $\ev(\mathbb{L})$ are the same as those we would obtain for $C(\X,\M)$, with $\M=\ini(\mathbb{L})$. Thus, we can restrict ourselves to monomial codes $C(\X,\M)$. 
\end{remark}

One can generalize these techniques to bound the RGHWs of evaluation codes \cite{hiramRGHW}. Let $\mathbb{L}_2\subset\mathbb{L}_1  \subset \text{Span}_\fq(\Delta(I(\X)))$ be two linear spaces, and, for  $1\leq r \leq \dim \mathbb{L}_1-\dim \mathbb{L}_2$, consider $\M_{1,2}^r:=\binom{\ini(\mathbb{L}_1\setminus \mathbb{L}_2)}{r}$. Let $M_{1,2}^r:=\binom{\varphi(\ini(\mathbb{L}_1\setminus \mathbb{L}_2))}{r}$. Then
\begin{equation}\label{eq:fprelative_ugly}
M_r(\ev(\mathbb{L}_1),\ev(\mathbb{L}_2))\geq \min \{\size{\nabla_X(M)} : M \in M_{1,2}^r\}.
\end{equation}
Following Remark \ref{r:evalcodes_to_monomial}, we usually restrict ourselves to monomial codes in this context. Let $\M_2\subset \M_1$, and we denote $M_i=\varphi(\M_i)$, for $1\leq i \leq 2$. Assume that $\alpha \in M_1\setminus M_2$ implies $\alpha \succ\beta$, for any $\beta \in M_2$. Then the previous bound can be rewritten as
\begin{equation}\label{eq:fprelative}
M_r(C(\X,\M_1),C(\X,\M_2))\geq \min \left\{\size{\nabla_X(M)}: M \in \binom{M_1\setminus M_2}{r}\right\}:=\FB_X^r(M_1,M_2).
\end{equation}

Similarly, for dual codes, we get
\begin{equation}\label{eq:fprelative_dual}
M_r(C(\X,\M_2)^\perp,C(\X,\M_1)^\perp)\geq \min \left\{\size{\nabla_X^\perp(M)}: M \in \binom{M_1\setminus M_2}{r}\right\}:=\FB_X^{r,\perp}(M_1,M_2).
\end{equation}

\begin{definition}
We say that a set $A\subset \NN^m$ is \textit{decreasing} (or \textit{downward closed}) if, for every $\alpha \in A$, we have $\set{\beta\in \NN^m :\beta \preceq \alpha}\subset A $. Similarly, we say that a set of monomials $\M\subset \mon $ is decreasing if the set $\varphi(\M)$ is decreasing. Given $U\subset X \subset \NN^m$, we say that it is \textit{increasing} (or \textit{upward closed}) if, for every $\alpha \in U$, we have $\set{\beta\in X :\alpha  \preceq \beta}\subset U $, i.e, if $X\setminus U$ is decreasing. 
\end{definition}

\begin{remark}
By construction, the sets $X$ and $\Delta(I(\X))$ from the previous section are decreasing. If we consider a code $C(\X,\M)$ where $\M$ is not decreasing, we obtain the same bound in \eqref{eq:footprint_bound_combinatorial} as if we considered $\M'$, the smallest decreasing set that contains $\M$, and the code $C(\X,\M')$ has a higher rate.

Similarly, if $x^\alpha\in \M$ and $x^{\alpha'}\not\in \M$, $\alpha'\preceq\alpha$, then $\M'=\M\setminus \{x^\alpha\}$ will give a code $C(\X,\M')^\perp$ with higher rate and same bound in \eqref{eq:footprint_bound_combinatorial_dual} as $C(\X,\M)^\perp$. 

Therefore, if we aim at maximizing \eqref{eq:footprint_bound_combinatorial} or \eqref{eq:footprint_bound_combinatorial_dual} individually, we may restrict ourselves to decreasing sets.  
\end{remark}

All the bounds we have given for $C(\X,\M)$ are expressed purely combinatorially in terms of $X=\varphi(\Delta(I(\X)))$ and $M=\varphi(\M)$. Thus, the analysis in the following sections will directly consider two decreasing sets $M\subset X\subset \NN^m$, rather than some particular $\M$ and $\X$. 

\section{Wei duality of the footprint bound}\label{s:wei_duality}
In this section, we show that both the bounds \eqref{eq:footprint_bound_combinatorial} and \eqref{eq:footprint_bound_combinatorial_dual} are ``Wei duals'' of each other, i.e., they satisfy a relation similar to that of Theorem \ref{t:ghwdual}. Therefore, the footprint bound is sharp if and only if the dual footprint bound is sharp. By Remark \ref{r:evalcodes_to_monomial}, without loss of generality, we may think about monomial codes. Thus we consider $M$, which is the set of exponents of the set of monomials we evaluate.

Assume that we have fixed $M\subset X\subset \NN^m$, with $X$ a decreasing set, and let $n=\size{X}$, $k=\size{M}$. We define now two sequences of numbers which are closely related to the bounds \eqref{eq:footprint_bound_combinatorial} and \eqref{eq:footprint_bound_combinatorial_dual}:
$$
\Gamma(w):=\max \set{\size{U\cap M}: U \text{ is increasing}, \; \abs{U}=w },
$$
$$
\Gamma^\perp(w):=\max \set{\size{V \cap (X\setminus M)}: V \text{ is decreasing}, \; \abs{V}=w }.
$$
We consider the sequences $\set{u_r}_{r=1}^k$ and $\set{u^\perp_r}_{r=1}^{n-k}$, where $w\in \set{u_r}_{r=1}^k$ if and only if $\Gamma(w)-\Gamma(w-1)=1$, and $w\in \set{u^\perp_r}_{r=1}^{n-k}$ if and only if $\Gamma^\perp(w)-\Gamma^\perp(w-1)=1$. Note that both $\Gamma$ and $\Gamma^\perp$ are increasing, and they can increase at most by 1 (you may take $U$ with $\size{U}=w$, and remove a minimal element to obtain another increasing $U'$ with $\size{U'}=w-1$, and thus $\Gamma(w-1)\geq \Gamma(w)-1$; similarly for $\Gamma^\perp$). 

\begin{lemma}\label{l:jumps}
Let  $M\subset X\subset \NN^m$ with $X$ a decreasing set, and let $1\leq r \leq \size{M}$, $1\leq r'\leq \size{X}-\size{M}$. We have
$$
u_r=\min \set{\size{U}: U \text{ is increasing,}\size{U\cap M}\geq r} \text{ and}
$$
$$
u_{r'}^\perp= \min \set{\size{V}: V \text{ is decreasing,}\size{V\cap (X\setminus M)}\geq r'}.
$$
\end{lemma}
\begin{proof}
By definition, we have
$$
u_r=\min\set{w: \Gamma(w)\geq r} \text{ and } u_{r'}^\perp=\min \set{w:\Gamma^\perp(w)\geq r'},
$$
since $\Gamma$ and $\Gamma^\perp$ increase by at most 1 when increasing $w$ by 1, which gives the result.
\end{proof}

\begin{lemma}\label{l:connection_FB}
Let  $M\subset X\subset \NN^m$ with $X$ a decreasing set. Then $u_r=\FB_X^r(M)$ and $u_{r'}^\perp=\FB_X^{r',\perp}(M)$, for any $1\leq r \leq \size{M}$ and $1\leq r'\leq \size{X}-\size{M}$. 
\end{lemma}
\begin{proof}
Consider $A\subset M$ with $\size{A}=r$ such that $\size{\nabla_X(A)}=\FB_X^r(M)$. Then $\nabla_X(A)$ is an increasing set, and $\size{\nabla_X(A)\cap M}\geq r$. Thus, $u_r\leq\size{\nabla_X(A)}=\FB_X^r(M)$ by Lemma \ref{l:jumps}. Now take $U$ increasing such that $\size{U}=u_r$, and $\size{U\cap M}\geq r$. We can take $A\subset U\cap M$ with $\size{A}=r$, and then $\nabla_X(A)\subset U$, i.e., $\FB_X^r(M)\leq \size{U}=u_r$. The proof for $u_{r'}^\perp$ is analogous.
\end{proof}

\begin{theorem}\label{t:duality_footprint}
Let  $M\subset X\subset \NN^m$ with $X$ a decreasing set. Then 
$$
\set{\FB_X^r(M)}_{r=1}^{\size{M}}\sqcup \set{n+1-\FB_X^{r',\perp}(M)}_{r'=1}^{\size{X}-\size{M}}
=\set{1,\dots,\size{X}}.
$$
\end{theorem}
\begin{proof}
Let $n=\size{X}$ and $k=\size{M}$. Consider $U$ an increasing set of size $w$. Then $V:=X\setminus U$ is a decreasing set with size $n-w$. We have 
$$
\size{V\cap M}=\size{V}-\size{V \cap (X\setminus M)} \text{ and }\size{U\cap M}=\size{M}-\size{V \cap M}.
$$
Substituting the first expression in the second one, we obtain
$$
\size{U\cap M}=k-n+w+\size{V \cap (X\setminus M)}.
$$
If we fix $w$, then $k-n+w$ is constant, and then
$$
\max_{\size{U}=w}\size{U\cap M}=k-n+w+\max_{\size{V}=n-w}\size{V \cap (X\setminus M)}.
$$
This also implies $\Gamma(w)=k-n+w+\Gamma^\perp(n-w)$. Taking differences, we obtain
$$
\Gamma(w)-\Gamma(w-1)=1+\Gamma^\perp(n-w)-\Gamma^\perp(n-w+1).
$$
Thus, since we know $\Gamma(w)-\Gamma(w-1)\in \set{0,1}$, we get $\Gamma(w)=\Gamma(w-1)+1$ if and only if $\Gamma^\perp(n-w)-\Gamma^\perp(n-w+1)=0$. In other words, $w\in \set{u_r}_{r=1}^{k}$ if and only if $n-w+1\not \in  \set{u_{r'}^\perp}_{r'=1}^{n-k}$. We finish the proof by Lemma \ref{l:connection_FB}.
\end{proof}

\begin{remark}
As a consequence of Theorem \ref{t:duality_footprint}, the footprint bound is sharp for all the GHWs if and only if the dual footprint bound is sharp for all the GHWs. 
\end{remark}

\begin{example}
In \cite{sanjose_simplex}, the authors prove that, for some evaluation codes defined over a simplex (introduced in \cite{kopparty_high_rate}), the footprint bound is sharp. Consequently, by Theorem \ref{t:duality_footprint}, the dual footprint bound is sharp for their dual codes. 
\end{example}

\begin{corollary}\label{c:monotonicity_fb}
Let  $M\subset X\subset \NN^m$ with $X$ a decreasing set. Then
$$
\FB_X^1(M)<\FB_X^2(M)<\cdots <\FB_X^{\size{M}}(M), 
$$
$$
\FB_X^{1,\perp}(M)<\FB_X^{2,\perp}(M)<\cdots <\FB_X^{\size{X}-\size{M},\perp}(M).
$$
\end{corollary}
\begin{proof}
This is a consequence of Theorem \ref{t:duality_footprint}, since non-strict monotonicity follows directly from the definitions.
\end{proof}

\begin{remark}
Corollary \ref{c:monotonicity_fb} can also be proven directly: given a set $B\in \binom{M}{r}$ such that $\size{\nabla_X(B)}=\FB_X^r(M)$, if we consider $B'\in \binom{M}{r-1}$ obtained by removing a minimal element of $B$, we get $\FB_X^{r-1}(M)\leq \size{\nabla_X(B')}\leq  \size{\nabla_X(B)}-1<\FB_X^r(M)$, and similarly for the dual footprint. Moreover, this also proves the strict monotonicity from Corollary \ref{c:monotonicity_fb} for the bounds from \eqref{eq:fprelative} and \eqref{eq:fprelative_dual}.
\end{remark}

Analogous results to Theorem \ref{t:duality_footprint} and Corollary \ref{c:monotonicity_fb} can be obtained for the Feng-Rao bound \cite{fengraoMajority} and the Andersen-Geil bound \cite{andersen_geil_bound}. Let $\mathcal{Y}$ be a smooth projective curve over $\fq$, let $P_1,\dots,P_n,Q$ be distinct $\fq$-rational points of
$\mathcal{Y}$, and let $D=P_1+\cdots+P_n$. Using the usual notation for one-point AG codes $C(D,\lambda Q)$, we consider the Weierstrass semigroup $H(Q)$ for $Q\in \mathcal{Y}$, and $H^*(Q)\subset H(Q)$ is the finite set of non-gaps where the dimension of the corresponding one-point AG code actually increases, i.e.,
$$
H^*(Q)=\set{\lambda \in H(Q):C(D,\lambda Q)\neq C(D,(\lambda-1)Q)}.
$$
Note that $H^*(Q)$ depends on $D$, and $\size{H^*(Q)}=n$. For each $\lambda \in H^*(Q)$, we choose $f_\lambda \in \mathcal{L}(\infty Q)$ with $-v_Q(f_\lambda)=\lambda$.  We may consider $M\subset H^*(Q)$. For $\alpha,\beta\in \NN$, we define the relation $\alpha \preceq_Q \beta$ if and only if $\beta-\alpha \in H(Q)$. For any $S\subset H^*(Q)$, we can consider $\nabla_Q(S)=\set{\beta \in H^*(Q):\exists \alpha \in S:\alpha \preceq_Q \beta}$ and $\nabla^\perp_Q(S)=\set{\alpha \in H^*(Q):\exists \beta \in S:\alpha \preceq_Q \beta}$. We denote
$$
\mathcal{C}_M(D,Q):=\text{Span}_{\fq}\set{\ev_{\set{P_1,\dots,P_n}}(f_\lambda):\lambda\in M}.
$$
We define the Andersen-Geil bound as in \eqref{eq:footprint_bound_combinatorial}, and the Feng-Rao bound as in \eqref{eq:footprint_bound_combinatorial_dual}, using the alternative definitions for $X$ ($H^*(Q)$ plays the role of $X$), $M$, $\preceq_Q$, $\nabla_Q$, and $\nabla_Q^\perp$ (similarly for the bounds for the RGHWs from \eqref{eq:fprelative} and \eqref{eq:fprelative_dual}), i.e., we define

\begin{equation}\label{eq:andersen_geil}
d_r(\mathcal{C}_M(D,Q))\geq \AG_Q^r(M):=\min\set{ \size{\nabla_Q(S)}:S\in \binom{M}{r}}, 
\end{equation}
\begin{equation}\label{eq:feng_rao}
d_r(\mathcal{C}_M(D,Q)^\perp)\geq  \FR_Q^{r}(M):=\min\set{ \size{\nabla^\perp_Q(S)}:S\in \binom{H^*(Q)\setminus M}{r}}.
\end{equation}

For $r=1$, these bounds have been shown to be consequences of each other in \cite{geil_feng_rao_decoding_primary}. In general, they also satisfy the following Wei-type duality and strict monotonicity.

\begin{theorem}\label{t:duality_fr_ag}
With the notation as above, we have
$$
\set{\AG_Q^r(M)}_{r=1}^{\size{M}}\sqcup \set{n+1-\FR_Q^{r'}(M)}_{r'=1}^{n-\size{M}}
=\set{1,\dots,n}.
$$
Moreover,
\[
\AG_Q^1(M)<\cdots<\AG_Q^{|M|}(M)
\]
and
\[
\FR_Q^{1}(M)<\cdots<
\FR_Q^{n-|M|}(M).
\]
\end{theorem}
\begin{proof}
It follows from the proofs of Theorem \ref{t:duality_footprint} and \ref{c:monotonicity_fb}, with
$X=H^*(Q)$ and with $\nabla_X$ and $\nabla_X^\perp$ replaced by
$\nabla_Q$ and $\nabla_Q^\perp$, respectively.
\end{proof}

\begin{example}
In \cite{sanjoseGHWNT}, the authors consider a modified footprint bound, which is sharp for decreasing norm-trace codes. Since the duals are monomially equivalent to decreasing norm-trace codes, one can use the same bound for the dual codes, and it is also sharp. In \cite{geil_secret_sharing_norm_trace}, the author shows that this footprint bound equals the Andersen-Geil bound for primary codes (or Feng-Rao for the dual codes). Thus, arguing as above, one can also deduce the sharpness of the Feng-Rao bound from Andersen-Geil's, or vice versa. 
\end{example}

In many well-known cases, the dual of a monomial evaluation code is monomially equivalent to another monomial evaluation code that uses the same evaluation points \cite{lopez_dual_evaluation_code,duursma_reed_muller_complete_intersections,sarabia_gorenstein_dual_codes}. The following result shows that, in those cases, the footprint bound associated with the monomials that define the dual coincides with the dual footprint bound.

\begin{proposition}\label{p:order_reversing}
Let $\phi:X\to X$ be an order-reversing bijection, meaning $\alpha\preceq \beta$ if and only if $\phi(\beta)\preceq \phi(\alpha)$. Let $M_2\subset M_1\subset X$ such that $\alpha \in M_1\setminus M_2$ implies $\alpha \succ\beta$, for any $\beta \in M_2$, and consider $M_i^\perp:=\phi(X\setminus M_i)$, for $i=1,2$. Then, for $1\leq r \leq \size{M_1}-\size{M_2}$, we have
$$
\FB_X^{r,\perp}(M_1,M_2)=\FB_X^{r}(M_2^\perp,M_1^\perp) \text{ and } \FB_X^{r,\perp}(M^\perp_2,M^\perp_1)=\FB_X^{r}(M_1,M_2).
$$
\end{proposition}
\begin{proof}
Let $B\in \binom{M_1\setminus M_2}{r}$, and consider $B^\perp:=\phi(B)\subset M_2^\perp\setminus M_1^\perp$. Since $\phi$ is order-reversing, we have $\phi(\nabla_X^\perp(B))=\nabla_X(B^\perp)$. Indeed, $\alpha\in \phi(\nabla^\perp_X(B))$ if and only if there is $\beta\in B$ and $\gamma$ such that $\alpha=\phi(\gamma)$ and $\gamma \preceq \beta$, i.e., we have $\phi(\beta) \preceq \alpha$, which is equivalent to $\alpha \in \nabla_X(B^\perp)$. Therefore, $\size{\phi(\nabla_X^\perp(B))}=\size{\nabla_X^\perp(B)}=\size{\nabla_X(B^\perp)}$, and we obtain the first equality by taking the minimum over all $B\in \binom{M_1\setminus M_2}{r}$. The second equality follows from $\phi(\nabla_X(B))=\nabla^\perp_X(B^\perp)$, which can be proven in an analogous way.
\end{proof}

For simplicity, in the next examples we assume $M_1=X$, and thus $\FB_X^{r,\perp}(X,M)=\FB_X^{r,\perp}(M)$, $\FB_X^{r}(M^\perp,\set{0})=\FB_X^{r}(M^\perp)$. 

\begin{example}
Let $\X=A_1\times \cdots \times A_m$ be a Cartesian product of sets $A_i\subset \fq$ of size $\size{A_i}=d_i$, for $1\leq i \leq m$. Then $X=\set{0,\dots,d_1-1}\times \cdots \times \set{0,\dots,d_m-1}$. Let $\phi:X\to X$ be the map defined by $\phi(\alpha):=(d_1-1,\dots,d_m-1)-\alpha$. It is clear that it is an order-reversing bijection. For any decreasing set of monomials $\M\subset \fq[x_1,\dots,x_m]$, consider $M=\varphi(\M)$ and $M^\perp:=\phi(X\setminus M)$ as above. It follows from \cite[Thm. 5.4]{lopez_dual_evaluation_code} that the dual of $C(\X,\M)$ is monomially equivalent to $C(\X,\M^\perp)$, where $\M^\perp:=\set{x^\alpha \in \fq[x_1,\dots,x_m]:\alpha \in \phi(X\setminus M)}$. By Proposition \ref{p:order_reversing}, we have that $\FB_X^{r,\perp}(M)=\FB_X^{r}(M^\perp)$, for any $1\leq r \leq \abs{X}-\abs{M}$. 
\end{example}

As with Proposition \ref{p:order_reversing}, if we consider a one-point AG code $\mathcal{C}_M(D,\lambda Q)$ such that its dual is of the form $\mathcal{C}_{M^\perp}(D,\lambda^\perp Q)$, and there exists $\phi:H^*(Q)\to H^*(Q)$ an order reversing bijection such that $M^\perp=\phi(H^*(Q)\setminus M)$, then 
\[
\FR_Q^r(M)=\AG_Q^r(M^\perp),
\qquad
1\leq r\leq n-|M|.
\]
A similar result follows for the bounds of the RGHWs of pairs of such codes, which are defined in an analogous manner to \eqref{eq:fprelative} and \eqref{eq:fprelative_dual}.

\begin{example}\label{ex:onepointcodes}
Consider a one-point AG code $C(D,\lambda Q)$ such that $C(D,\lambda Q)^\perp\simeq C(D,(\mu -\lambda )Q)$, up to monomial equivalence, and consider $\phi(\alpha)=\mu+1-\alpha$. If we assume 
\[
\phi(H^\ast(Q))=H^\ast(Q), \text{ or, equivalently, } \alpha\in H^\ast(Q) \iff \mu+1-\alpha\in H^\ast(Q),
\]
then \(\phi\) is an order-reversing bijection. Consequently, if $M:=\{\alpha\in H^\ast(Q):\alpha\le \lambda\}$ and  $M^\perp:=\{\alpha\in H^\ast(Q):\alpha\le \mu -\lambda\}$, then $M^\perp=\phi\bigl(H^\ast(Q)\setminus M\bigr)$. It follows that
\[
\FR_Q^r(M)=\AG_Q^r(M^\perp)
\]
for $1\leq r\leq n-|M|$.

The condition $H^\ast(Q)=\mu+1-H^\ast(Q)$ is satisfied in some well-known cases. For example, if $D$ corresponds to the sum of all the affine rational points of the Hermitian curve and $Q$ corresponds to the point at infinity, we have
\[
H^\ast(Q)=\left\{iq+j(q+1):0\le i\le q^2-1,\; 0\le j\le q-1\right\}.
\]
Then we know that $C(D,\lambda Q)^\perp\simeq C(D,(q^3+q^2-q-2-\lambda)Q)$. If we take $\mu=q^3+q^2-q-2$, then we can check that $H^\ast(Q)=\mu+1-H^\ast(Q)$ (note that $\mu+1=(q^2-1)q+(q-1)(q+1)$), and the map $\phi(\alpha)=\mu+1-\alpha$ is an order-reversing bijection. 
\end{example}

The condition $H^\ast(Q)=\mu+1-H^\ast(Q)$ is closely related to the conditions required to obtain isometry dual flags of one-point AG codes. A complete flag is a sequence of codes $\set{0}=C_0\subsetneq C_1\subsetneq \cdots \subsetneq C_n= \fq^n$, and its dual flag is $\set{0}=C_n^\perp \subsetneq C_{n-1}^\perp \subsetneq\cdots \subsetneq C_0^\perp=\fq^n$. A flag is isometry-dual if there is an isometry $g$  such that $g(C_i)=C^\perp_{n-i}$, for $0\leq i \leq n$. The condition $H^\ast(Q)=\mu+1-H^\ast(Q)$ implies that the associated flag is isometry dual (assuming that the dual of $C(D,\lambda Q)$ is as in Example \ref{ex:onepointcodes}). For example, if $n\geq  2g+2$, one can take $\mu=n+2g-2$ (see \cite{geil_on_order_bound_ag,bras_isometry_dual}).

\section{Asymptotic analysis}\label{s:asymptotic}
We devote this section to proving that the footprint bound is insufficient to guarantee that a family of evaluation codes is asymptotically good, in the following sense.

\begin{definition}
Given a sequence of codes $\{C_n\}_{n=1}^\infty$, we say that it is \textit{asymptotically good} if we have
$$
\liminf_{n\to \infty }\frac{\dim(C_n)}{n}>0 \text{ and }\liminf_{n\to \infty }\frac{d_1(C_n)}{n}>0.
$$
\end{definition}

We start with a technical lemma, which will lead to Theorem \ref{t:asymptotically_bad_lattices}, the main result of the section.

\begin{lemma}\label{l:key_lemma_asymptotic}
For any decreasing set $X\subset \bqm$ with $\abs{X}=n_X\geq 2$, we have
\begin{equation}\label{eq:lemma_gamma}
\sum_{z\in X}\prod_{i=1}^m (z_i+1)\leq n_X^\gamma,
\end{equation}
where $\gamma=\log_q\left(\frac{q(q+1)}{2}\right)$. 
\end{lemma}
\begin{proof}
We argue by induction on $m$. Let $m=1$. Then $X=\{0,\dots,n_X-1\}$, and we need to show that 
\begin{equation}\label{eq:base_case_inducion}
\sum_{z\in X}\prod_{i=1}^m (z_i+1)=\frac{n_X(n_X+1)}{2}\leq n_X^\gamma. 
\end{equation}
By taking base-$n_X$ logarithms, it is enough to show that $f(x)=\log_x\left(\frac{x(x+1)}{2}\right)\leq \gamma=f(q)$, for all $2\leq x \leq q$ (recall that $X\subset [0,q-1]$, and $n_X\leq q$). This follows from the fact that the function is strictly increasing for $x> 1$.

Now we assume $m>1$. For each $j\in \{0,\dots,q-1\}$, consider
$$
X_j:=\{z\in [0,q-1]^{m-1}:(z,j)\in X\}. 
$$
By adding elements of $X$ with a fixed last coordinate, we can write
\begin{equation}\label{eq:eq_random}
\sum_{(z,j)\in X}\prod_{i=1}^{m-1} (z_i+1)(j+1)=\sum_{j=0}^{q-1}(j+1)\left( \sum_{z\in X_j}\prod_{i=1}^{m-1}(z_i+1) \right). 
\end{equation}
For every $j$ such that $\abs{X_j}\geq 2$, the induction hypothesis
gives
\[
\sum_{z\in X_j}\prod_{i=1}^{m-1}(z_i+1)
\leq n_{X_j}^{\gamma}.
\]
The same inequality is immediate when $\abs{X_j}\leq 1$. Hence,
summing in \eqref{eq:eq_random}, we obtain
\begin{equation}\label{eq:first_inequality}
\sum_{z\in X}\prod_{i=1}^m(z_i+1)
\leq
\sum_{j=0}^{q-1}(j+1)n_{X_j}^{\gamma}.
\end{equation}
Note that $n_X=\sum_{j=0}^{q-1} n_{X_j}$, which follows from the fact that $X=\bigsqcup_{j=0}^{q-1} X_j\times \{j\}$. Thus, to finish the proof, we have to prove
$$
\sum_{j=0}^{q-1}(j+1)n_{X_j}^\gamma\leq n_X^\gamma =  \left(\sum_{j=0}^{q-1} n_{X_j}\right)^\gamma .
$$
We denote $p_j:=n_{X_j}/n_X$. Since $X_{0}\supset X_{1}\supset \cdots \supset X_{q-1}$, we have $p_0\geq p_1\geq \cdots \geq p_{q-1}\geq 0$ and $\sum_{j=0}^{q-1}p_j=1$, and we have to show that
$$
\sum_{j=0}^{q-1}(j+1)p_j^\gamma \leq 1. 
$$
Consider $g(p_0,\dots,p_{q-1}):=\sum_{j=0}^{q-1}(j+1)p_j^\gamma $, which is strictly convex. Indeed, since $\gamma>1$, we have that $x^\gamma$ is strictly convex in $[0,\infty)$ (this can be checked using derivatives and the definition for $x=0$), and thus $g$ is strictly convex because it is a sum of strictly convex functions. Therefore, we are trying to find the maximum of a strictly convex function over a closed convex polytope $\mathcal{P}$ (the standard simplex, with the extra condition $p_0\geq p_1\geq \cdots \geq p_{q-1}$). This implies that the maximum is attained at one of the vertices of $\mathcal{P}$. The vertices of $\mathcal{P}$ are the points $P_k$ with $p_0=\cdots=p_{k-1}=1/k$, and $p_j=0$ for $j=k,\dots,q-1$ (this can be seen, for example, by considering the change of variables $z_k=(k+1)(p_k-p_{k+1})$, $0\leq k \leq q-2$, and $z_{q-1}=q p_{q-1}$, which maps $\mathcal{P}$ to the standard simplex). 
We have
\begin{equation}\label{eq:g_of_p}
g(P_k)=\sum_{j=0}^{k-1}(j+1)\frac{1}{k^\gamma}=\frac{k(k+1)}{2k^\gamma}\leq 1.
\end{equation}
For $2\leq k\leq q$, the last inequality follows from the
one-dimensional case \eqref{eq:base_case_inducion}, while
$g(P_1)=1$. Thus,
\[
g(p_0,\dots,p_{q-1})\leq1
\]
for every point in $\mathcal P$.
\end{proof}

\begin{theorem}\label{t:asymptotically_bad_lattices}
Let $\{M_m\}_{m=1}^\infty$, $\{X_m\}_{m=1}^\infty$ be two sequences of sets such that $\emptyset \neq M_m\subset X_m\subset \bqm$ and $X_m$ is decreasing, for $m\geq 1$. Assume that $\lim_{m\to \infty}\abs{X_m}=\infty $. Then 
$$
\lim_{m\to \infty}\frac{\abs{M_m}\FB^1_{X_m}(M_m)}{\abs{X_m}^2}=0.
$$
As a consequence, if both $\lim_{m\to \infty} \abs{M_m}/\abs{X_m}$ and $\lim_{m\to \infty} \FB^1_{X_m}(M_m)/\abs{X_m}$ exist, then either $\lim_{m\to \infty} \abs{M_m}/\abs{X_m}=0$ or $\lim_{m\to \infty} \FB^1_{X_m}(M_m)/\abs{X_m}=0$.
\end{theorem}
\begin{proof}
Let $n_m:=\abs{X_m}$, and denote $\abs{M_m}=a_m n_m$, $\FB^1_{X_m}(M_m)=b_m n_m$. Consider
$$
S:=\sum_{P\in M_m}\FB^1_{X_m}(P).
$$
We clearly have $S\geq \abs{M_m}\FB^1_{X_m}(M_m)=a_mb_m n_m^2$, and also
\begin{equation}\label{eq:bound_on_S}
S=\sum_{P\in M_m} \sum_{\substack{z \in X_m \\ z \succeq P}} 1\leq \sum_{P\in \bqm} \sum_{\substack{z \in X_m \\ z \succeq P}} 1 = \sum_{z \in X_m} \sum_{\substack{P\in \bqm \\ P \preceq z}} 1 = \sum_{z \in X_m} \prod_{i=1}^m(z_i+1)\leq n_m^\gamma,
\end{equation}
where we have used Lemma \ref{l:key_lemma_asymptotic} for the last inequality. Thus, we have
$$
a_mb_m n_m^2\leq S \leq n_m^\gamma,
$$
which implies 
\begin{equation}\label{eq:n_gamma}
a_mb_m\leq n_m^{\gamma-2}.
\end{equation}
Note that $\gamma<2$, since $q(q+1)/2<q^2$. Thus, the right-hand side of \eqref{eq:n_gamma} tends to 0 when $m\to \infty$, forcing $\lim_{m\to \infty}a_mb_m =0$. 
\end{proof}

\begin{corollary}
Let $\mathbb{L}_m\subset \fq[x_1,\dots,x_m]/I(\X_m)$ be a sequence of linear subspaces, with respect to a sequence of sets of points $\X_m\subset \fq^m$. Let $L_m=\varphi(\ini(\mathbb{L}_m))$ and $X_m=\varphi(\Delta(I(\X_m)))$. Let $k_m:=\abs{L_m}/\abs{X_m}$, $\delta_m:=\FB^1_{X_m}(L_m)/\abs{X_m}$, and assume $\lim_{m\to\infty}\size{X_m}=\infty$. Then it is not possible to have both
$$
\liminf_{m\to \infty }k_m>0 \text{ and }\liminf_{m\to \infty }\delta_m>0.
$$
In other words, the footprint bound cannot be used to guarantee asymptotic goodness. 
\end{corollary}
\begin{proof}
It follows from Theorem \ref{t:asymptotically_bad_lattices}, also considering Remark \ref{r:evalcodes_to_monomial} to reduce to the monomial case. 
\end{proof}

Note that these results do not imply that the family of monomial codes, or even decreasing monomial codes, is bad. 
However, they imply that, to guarantee asymptotic goodness, we need to use other bounds, or improvements of the footprint, such as the one used in \cite{sanjoseGHWNT}. 
We now introduce a new notion closely related to that of asymptotic goodness (see \cite{geilRampSecret}, where this idea was studied in the context of asymptotically good secret sharing schemes).

\begin{definition}
Let $r\geq 1$. We say that a sequence of codes $C_n\subset \fq^n$ is $r$-\textit{asymptotically good} if we have both 
$$
\liminf_{n\to \infty }\frac{\dim(C_n)}{n}>0 \text{ and }\liminf_{n\to \infty }\frac{d_r(C_n)}{n}>0.
$$
For the second quotient to make sense, we must have $\dim(C_n)\geq r$ (we could also define $d_r(C)=0$ if $r>\dim(C)$). This is guaranteed to happen for $n$ sufficiently large if we have $\liminf_{n\to \infty }\frac{\dim(C_n)}{n}>0$.
\end{definition}
By the definition, it is clear that if a sequence of codes is $r$-asymptotically good, it will be $j$-asymptotically good for every $j\geq r$. We can also obtain sequences which are  $r$-asymptotically good but not $r-1$-asymptotically good.

\begin{lemma}
Let $r\geq 2$ and let $\{C_n\}_{n=1}^\infty$ be an asymptotically good sequence of codes. Consider the family defined by $C'_n=C_n$ for $n<r-1$, and $C'_n=C_n+\langle e_1,\dots,e_{r-1}\rangle$ for $n\geq r-1$. The sequence $\{C'_n\}_{n=1}^\infty$ is $r$-asymptotically good, but not $r-1$-asymptotically good.
\end{lemma}
\begin{proof}
Since $\{C_n\}_{n=1}^\infty$ is asymptotically good, for a sufficiently large $n$, we have $d_r(C_m)> d_1(C_m)>r$, and $\dim(C_m)\geq r$, for every $m\geq n$. Consider $D\subset C'_m$, a subcode with $\dim D=r$. By the definition of $C'_m$, we have that $D$ contains $c\in C'_m\setminus \langle e_1,\dots,e_{r-1}\rangle$, and we also have $\wt(c)\geq d_1(C_m)-(r-1)$. Thus, $d_r(C'_m)\geq \abs{\supp(D)}\geq d_1(C_m)-(r-1)$. This proves that $\{C'_n\}_{n=1}^\infty$ is $r$-asymptotically good. That $\{C'_n\}_{n=1}^\infty$ is not $r-1$-asymptotically good follows from the fact that, for sufficiently large $n$, we have $d_{r-1}(C'_n)=r-1$, since $C'_n$ contains the subcode $\langle e_1,\dots,e_{r-1}\rangle$. 
\end{proof}

Even though the footprint bound is not sufficient to prove asymptotic goodness, one might think it could still be sufficient to prove $r$-asymptotic goodness; however, in the next result we prove that this is not the case. 

\begin{corollary}\label{c:fb_r_bad}
Let $r\geq 1$. Under the same hypotheses of Theorem \ref{t:asymptotically_bad_lattices}, and, additionally, $|M_m|\geq r$ for all sufficiently large $m$, we have
$$
\lim_{m\to \infty}\frac{\abs{M_m}\FB^r_{X_m}(M_m)}{\abs{X_m}^2}=0.
$$
As a consequence, the footprint bound cannot be used to guarantee $r$-asymptotic goodness, for any $r\geq 1$. 
\end{corollary}
\begin{proof}
Let $n_m = \abs{X_m}$ and $\abs{M_m} = a_m n_m$. Let $P_1, \dots, P_r$ be the $r$ elements in $M_m$ with the smallest individual footprints $\FB^1_{X_m}(P_i)$. By the definition of the footprint bound
$$
\begin{aligned}
\FB^{r}_{X_m}(M_m)\leq \FB^{r}_{X_m}(\{P_1,\dots,P_r\})  \leq \left| \bigcup_{i=1}^r \{x \in X_m \mid x \succeq P_i\} \right| &\leq \sum_{i=1}^r \FB^1_{X_m}(P_i)\\
&\leq \frac{r}{\abs{M_m}} \sum_{P \in M_m} \FB^1_{X_m}(P),
\end{aligned}
$$
where in the last step we have used that the sum of the smallest footprints is lower than $r$ times the average footprint. Using \eqref{eq:bound_on_S} we obtain
$$
\FB^{r}_{X_m}(M_m) \leq\frac{r}{\abs{M_m}} \sum_{P \in M_m} \FB^1_{X_m}(P)\leq r \frac{n_m^\gamma}{a_m n_m}.
$$
Multiplying both sides by $a_m / n_m$, we get
$$
\frac{\abs{M_m}\FB^{r}_{X_m}(M_m)}{\abs{X_m}^2} \leq r n_m^{\gamma - 2}.
$$
Since $r$ is fixed and $\gamma < 2$, the right-hand side tends to 0 as $m \to \infty$. 
\end{proof}

We can also consider the relative footprint bound and a sequence of nested sets (or codes), similar to \cite{geilRampSecret}, but we obtain the same negative result.

\begin{corollary}\label{c:fb_r_bad_relative}
Let $r\geq 1$. Let $\{M^1_m\}_{m=1}^\infty$, $\{M^2_m\}_{m=1}^\infty$, $\{X_m\}_{m=1}^\infty$ be three sequences of sets such that $\emptyset \neq M^2_m\subset M^1_m\subset X_m\subset \bqm$ and $X_m$ is decreasing, for $m\geq 1$. 
Assume that $\alpha \in M_m^1\setminus M_m^2$ implies $\alpha \succ\beta$, for any $\beta \in M_m^2$, and that $\lim_{m\to \infty}\abs{X_m}=\infty $. Also assume that $|M_m^1\setminus M_m^2|\geq r$ for all sufficiently large $m$. Then 
$$
\lim_{m\to \infty}\frac{\abs{M^1_m\setminus M^2_m}\FB^r_{X_m}(M^1_m,M^2_m)}{\abs{X_m}^2}=0.
$$
\end{corollary}
\begin{proof}
Let $n_m = \abs{X_m}$ and $\abs{M^1_m\setminus M^2_m} = a_m n_m$. The proof is then analogous to that of Corollary \ref{c:fb_r_bad}, considering $D_m:=M^1_m\setminus M^2_m$ instead of $M_m$, and relative footprints. 
\end{proof}

If we are considering the setting in which we do not have the condition that $\alpha \in M_m^1\setminus M_m^2$ implies $\alpha \succ\beta$, then we have to consider the bound from \eqref{eq:fprelative_ugly}. Although that bound does not have a purely combinatorial description, we may still consider $\varphi(\ini(\mathbb{L}_m^1\setminus \mathbb{L}_m^2))$ instead of $M_m^1\setminus M_m^2$ in the proof of Corollary \ref{c:fb_r_bad_relative}, and an analogous result follows. 

\begin{remark}
Note that the proof of Corollaries \ref{c:fb_r_bad} and \ref{c:fb_r_bad_relative} holds even if $r$ is not a constant, as long as $r=o(n_m^{2-\gamma})$.
\end{remark}

Similar arguments show that the bounds for the RGHWs of the dual codes are also insufficient to obtain asymptotically good codes. Indeed, Lemma \ref{l:key_lemma_asymptotic} is already written appropriately for the dual footprint bound; see also the last equality in \eqref{eq:bound_on_S}.  

\begin{example}
The footprint bound is sharp for codes considered in \cite{sanjose_simplex}, which implies they are asymptotically bad. Similarly, Theorem \ref{t:asymptotically_bad_lattices} recovers that hyperbolic codes are asymptotically bad, see \cite[Thm. 2]{olav_codes_from_order_domains}.
\end{example}

While the value of the footprint bound may depend on the chosen monomial ordering, which determines $X=\varphi(\Delta(I(\X)))$ in terms of $\X$, once we have chosen that ordering, that fixes $X$, and the previous results apply, i.e., it will still satisfy the Wei-type duality and the asymptotic badness. 

\section{Conclusion}
In this paper, we have proven a Wei-type duality for the footprint bound and, more generally, for the Andersen-Geil and the Feng-Rao bounds. We have also obtained that the footprint bound cannot be used to prove that a family of codes is asymptotically good. This motivates the study of modified footprints as in \cite{sanjoseGHWNT,geilbezout,olav_codes_from_order_domains,geil_improvement_feng_rao_primary,geil_affine_variety_klein}, which may be able to bypass this restriction. It would also be interesting to see whether it is possible to prove similar results for generalizations of the footprint bound to projective and weighted projective spaces.


\end{document}